\documentclass[11pt,a4paper]{article}

\usepackage{graphicx}
\usepackage{amssymb,amsmath,amstext,amsfonts}

\newtheorem{theorem}{\sc Theorem}
\newtheorem{lemma}{\sc Lemma}
\newtheorem{proposition}{\sc Proposition}
\newtheorem{prope}{\sc Property}
\newtheorem{coro}{\sc Corollary}
\newtheorem{nota}{\sc Notation}
\newtheorem{defin}{\sc Definition}
\newtheorem{cla}{\sc Claim}
\newtheorem{rem}{\sc Remark}
\newenvironment{proof}{\par \sc Proof.\rm}{\hspace*{\fill}$\Box$\vspace{1ex}}

\begin{document}

\title{Tree Nash Equilibria in the \\ Network Creation Game}
\author{Akaki Mamageishvili \and Mat\'u\v{s} Mihal\'ak \and
Dominik M\"uller}
\maketitle

\begin{abstract}
  In the network creation game with $n$ vertices, every vertex (a player) buys a
  set of adjacent edges, each at a fixed amount $\alpha>0$. 
  %
  %
  It has been conjectured that for $\alpha \geq n$, every Nash equilibrium is a
  tree, and has been confirmed for every $\alpha \geq 273\cdot n$.
  We improve upon this bound and show that this is true for every $\alpha \geq
  65\cdot n$.
  To show this, we provide new and improved results on the local structure of
  Nash equilibria. Technically, we show that if there is a cycle in a Nash
  equilibrium, then $\alpha < 65\cdot n$. Proving this, we only consider
  relatively simple strategy changes of the players involved in the cycle.
  We further show that this simple approach cannot be used to show the desired
  upper bound $\alpha < n$ (for which a cycle may exist), but conjecture that a
  slightly worse bound $\alpha < 1.3 \cdot n$ can be achieved with this
  approach. Towards this conjecture, we show that if a Nash equilibrium has a
  cycle of length at most 10, then indeed $\alpha < 1.3\cdot n$.
  We further provide experimental evidence suggesting that when the girth of
  a Nash equilibrium is increasing, the upper bound on $\alpha$ obtained by the
  simple strategy changes is not increasing.
  To the end, we investigate the approach for a coalitional variant of Nash
  equilibrium, where coalitions of two players cannot collectively improve, and
  show that if $\alpha\geq 41\cdot n$, then every such Nash equilibrium is a
  tree.
%
\end{abstract}


\section{Introduction}

Network creation game has been introduced by Fabrikant et
al.\,\cite{Fabrikant+etal/2003} as a formal model to study the effects of
strategic decisions of economically motivated agents in decentralized networks
such as the Internet. In such networks, local decisions including those about
infrastructure are decided by autonomous systems. Autonomous systems follow
their own interest, and as a result, their decisions may be sub-optimal for the
whole society. 
Network creation games allow to formally study the structure of networks created
in such a manner, and to compare them with potentially optimal networks (optimal
with respect to the whole society).
In the network creation game, there are $n$ players $V=\{1,\ldots,n\}$, each
representing a vertex of an undirected graph. The strategy $s_i$ of a player $i$
is to create (or buy) a set of adjacent edges, each at a fixed amount $\alpha >
0$. The played strategies $s=(s_1,\ldots,s_n)$ collectively define an edge-set
$E_s$, and thus a graph $G_s=(V,E_s)$. The goal of every player is to minimize
its cost $c_i$, which is the amount paid for the edges (creation cost), plus the
total distances of the player to every other node of the resulting network $G$
(usage cost), i.e., 
\begin{displaymath}
  c_i(s) = \alpha \cdot |s_i| + \sum_{j=1}^{n} \text{dist}(i,j)\text{,}
\end{displaymath}
where $\text{dist}(i,j)$ denotes the distance between $i$ and $j$ in the
resulting network $G$.

A strategy vector $s=(s_1,\ldots,s_n)$ is a \emph{Nash equilibrium} if no player
$i$ can change the set $s_i$ of created edges to another set $s_i'$ and improve
its cost $c_i$. Abusing the definition, the resulting graph $G_s$ itself is
called a Nash equilibrium, too, and we define its (social) cost $c(G)$ to be
the cost $c(s)$, i.e., the cost of the corresponding strategy vector $s$.
The \emph{social cost} $c(s)$ of strategy vector $s$ is the sum of the
individual costs, i.e., $c(s)=\sum_{i=1}^n c_i(s)$.
It is a trivial observation to see that in any Nash equilibrium $G_s$, no edge is
bought more than once. From now on, we only consider such strategy vectors, and
observe then that 
\begin{displaymath}
  c(s) := \sum_{i=1}^n c_i(s) = \alpha\cdot |E_s| + \sum_{i=1}^n \sum_{j=1}^n
  d(i,j)\text{.}
\end{displaymath}

A graph $G=(V,E)$ can be created by many strategy vectors $s$ (precisely in
$2^{|E|}$ many ways, because every edge in $E$ can be bought by exactly one of its
endpoints), but each of such realizations has the same social cost. 
Graph $G^*=(V,E)$ is an \emph{optimum} graph, if it minimizes the \emph{social cost}
$c(s)$ (for any strategy vector $s$ for which $G_s = G$).

Let $\mathcal{N}$ denote the set of all Nash equilibria of a network creation
game on $n$ vertices and edge-price $\alpha$. The \emph{price of anarchy} (PoA)
of the network creation game is the ratio 
\begin{displaymath}
  \text{PoA} = \max_{s\in \mathcal{N}} \frac{c(G_s)}{c(G^*)}\text{.}
\end{displaymath}
Price of anarchy expresses the (worst-case) loss of the quality of a network that
the society could achieve.

In a series of papers
\cite{Fabrikant+etal/2003,Albers+etal/2006,Demaine+etal/2012,Mihalak+Schlegel/2013} it has been shown that the price of anarchy of the
network creation game is $O(1)$, i.e., a constant independent of both $n$ and
$\alpha$, for every value $\alpha>0$ with the exception of the range
$n^{1-\varepsilon} < \alpha < 273\cdot n$, where $\varepsilon =\Omega(\frac{1}{\log
n})$.
For the value of $\alpha$ with $n^{1-\varepsilon} < \alpha < 273\cdot n$, an
upper bound of $2^{\sqrt[n]{\log n}}$ on the price of anarchy is known (while no
Nash equilibrium with considerably large social cost is known). 
It is conjectured, however, that the price of anarchy is constant also in this
range of $\alpha$. It remains a major open problem to confirm or disprove this
conjecture.
It is certainly of interest to note that there are several variants of the
network creation game (see, e.g.,
\cite{Alon+etal/2013,Brautbar+Kearns/2011,Ehsani+etal/2011,Bilo+Guala+Proietti/2012}),
but in none of these, with the exception of \cite{Demaine+Zadimoghaddam/2010},
the price of anarchy could be shown to be constant.

Understanding the structure of Nash equilibria has proven to be important in
bounding the price of anarchy. Fabrikant et al.\,\cite{Fabrikant+etal/2003}
showed that the social cost of any tree $G$ in Nash equilibrium is upper-bounded
by $O(1)\cdot c(G^*)$. 
Therefore, the price of anarchy is $O(1)$ for all values of $\alpha$ for which
every Nash equilibrium is a tree.
It has been shown that every Nash equilibrium is a tree for all values of
$\alpha$ greater than $n^2$, $12n\log n$, and $273 n$, respectively, in
\cite{Fabrikant+etal/2003},\cite{Albers+etal/2006}, and
\cite{Mihalak+Schlegel/2013}.
It has been conjectured that every Nash equilibrium is a tree for every
$\alpha\geq n$. Since for $\alpha=n/2$, non-tree Nash equilibria are known, this
\emph{tree conjecture} is asymptotically tight. 

In this paper, we make steps in the direction of resolving the tree conjecture. 
We first tighten the tree conjecture and provide a construction of a non-tree
Nash equilibrium for every $\alpha = n-3$ (thus, showing that, asymptotically,
one cannot hope to show that every Nash equilibrium is a tree for some value
$\alpha < n$).
We then apply a ``linear-programming-like'' approach to show that for $\alpha
\geq 65n$, every Nash equilibrium is a tree. To show this, we obtain new
structural results on Nash equilibria and combine them with the previous
approach of \cite{Mihalak+Schlegel/2013}.
Towards the end, we make further steps towards the conjecture. 
We show that if $\alpha \geq n$, then there is no non-tree Nash equilibrium
containing exactly one cycle. 
We then apply the ``linear-programming-like'' approach again to show that the
girth of every non-tree Nash equilibrium (for any $\alpha \geq n$) is at least
6.
Using the same ideas, we show that if a non-tree Nash equilibrium has girth at
most 10, then $\alpha \leq 1.3 n$. By further experimental results, we
conjecture that this holds for any girth, i.e., that non-tree Nash equilibria
can appear only for $\alpha \leq 1.3 n$.
%


\section{Preliminaries}

In the following, we will often denote the considered Nash equilibrium graph
$G_s=(V,E_s)$ of a network creation game with $\alpha>0$ simply as $G=(V,E)$.
Even though the graph $G_s$ is undirected, we will often direct the edges to
express the identity of the player which bought the edge in $s$; An edge
$\{u,v\}$ directed from $u$ to $v$ denotes the fact that $u$ bought/created the
edge in $s$.

Every non-tree $G$ contains a cycle. Let $c$ be the length of a shortest cycle
$C$ in $G$, and let $a_0,a_1,\ldots,a_{c-1}$ be the players that form one such
shortest cycle, and where $\{a_i,a_{i+1}\}\in E$ for every
$i=0,1,\ldots,c-1$ (where indices on vertices of the cycle are in the whole
paper to be understood modulo $c$). Observe the crucial property of a shortest
cycle $C$: the distance between $a_i$ and $a_j$ in the graph $G$ is equal to the
distance between $a_i$ and $a_j$ on the cycle $C$.

We will consider the players on the cycle $C$ and their strategy-changes that
involve only the $c$ edges of the cycle. For each strategy-change $s_{a_i}'$ of
player $a_i$, we obtain an inequality $c_i(s) \leq
c_i(s_1,\ldots,s_{a_i}',\ldots,s_n)$ stating simply the fact that in a Nash
equilibrium $s$, player $a_i$ cannot improve by changing its strategy.
We will often express such an inequality in the form of ``SAVINGS'' $\leq$
``INCREASE'', where ``SAVINGS'' denotes the parts of $c_i(s)$ that decreased
their value in $c_i(s')$, and ``INCREASE'' denotes the parts of $c_i(s)$ that
increased their value in $c_i(s')$.
For example, assume that $a_i$ buys the edge $e=\{a_i,a_{i+1}\}$ (i.e., $e\in
s_i$), and let us consider the strategy change where $a_i$ \emph{deletes} the
edge $e$ (i.e., $s_i' = s_i\setminus\{e\}$). Recall that $c_i(s) =
\alpha\cdot|s_i| + \sum_j d(i,j)$. Then, in such a strategy change, the
``SAVINGS'' are clearly on the edge-creation side, i.e., the player $a_i$ saves
$\alpha$ for not paying for the edge $e$. At the same time, some distances of
player $i$ may have increased -- the distance to a vertex $v$ increases, if in
$G_s$ every shortest path from $a_i$ to $v$ uses the deleted edge $e$. 
But the distance to $v$ could have increased by at most $c-2$ (as before, $a_i$
needed to go to vertex $a_{i+1}$ but now the vertex $a_{i+1}$ can be reached
``around'' the cycle).
Because of the Nash equilibrium property of $s$, we have ``SAVINGS'' $\leq$
``INCREASE'', which implies $\alpha \leq (c-2)(n-1)$ (as the distance to at
most $n-1$ vertices could have increased).

In the following, we will use slightly more involved forms of the just
described inequalities. For that reason, we will partition the vertices
according to their distances to the vertices from the cycle.
Let us fix a vertex $v\in V$. Let $G\setminus C$ be the graph $G$ without the
$c$ edges of the cycle $C$. Let us denote the distances of $v$ to the vertices
$a_0,a_1,\ldots,a_{c-1}$ in $G\setminus C$ by the vector
$d(v)=(d_0,d_1,\ldots,d_{c-1})$, respectively, where $d_i=\infty$ if $a_i$ and
$v$ are disconnected in $G\setminus C$. We call $d_i$ the \emph{outer distance}
of $v$ to $a_i$ in the Nash equilibrium $G$, and $d$ the \emph{vector of outer
distances} of $v$ in $G$. We now partition the vertices of $V$ by
this vector of outer distances. We will coarsen the partition in the
following way.
Observe that $d_s(a_i,v)$ in $G_s$ is now equal to $\min_j (d_s(a_i,a_j) + d_j)$,
because there always is a shortest path from $a_i$ to $v$ that first uses a part
of the cycle $C$ (until vertex $a_j$), leaves $C$ and never comes back to $C$.
Therefore, $\min_j d_j \leq d_s(a_i,v) \leq (c-1) + \min_j d_j)$. Moreover, for
any strategy change $s_i'$ of player $a_i$ which leaves $a_i$ connected by an
edge to a vertex of $C$, we still have $\min_j d_j \leq d_{s'}(a_i,v)\leq (c-1)
+ \min_j d_j$ (because there is a path from $a_i$ to the vertex $a_j$ of
smallest entry $d_j$ using the edge and the remaining of the cycle).
Because we are interested in the changes of the distances from $a_i$, i.e., in
the value of $\Delta:=d_{s'}(a_i,v) - d_s(a_i,v)$, we can normalize the vector
$d(v)$ by subtracting $\min_j d_j$ from each of the elements
$d_0,d_1,\ldots,d_{c-1}$ (which does not change the value of $\Delta$).
Observe that after the normalization, there is an entry $d_i$ equal to zero.
We will ``normalize'' the entries further more. Since we are interested in the
value $\Delta$, we can handle all entries $d_j\geq c-1$ in the same way: they do
not have any influence on $\Delta$ at all (no shortest path from vertex $a_i$,
$i\neq j$, will ever use $a_j$ to reach vertex $v$). We will therefore further
modify the vector $d$ by substituting every entry $d_j \geq c-1$ with the value
$c-1$. 

\begin{figure}[t]
  \centering
  \includegraphics[width=0.4\textwidth]{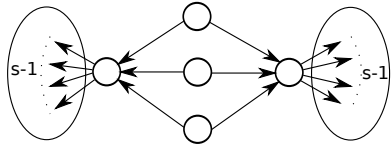}
  \caption{Non-tree Nash equilibrium for $n=2s+3$ players and $\alpha=n-3$. An
  edge directed from a node $u$ to a node $v$ denotes that $u$ buys the edge.}
  \label{fig:non-tree_NE}
\end{figure}

This gives partition of all vertices into groups $V_d$, where each group has
associated vector of ``normalized'' outer distances $d=(d_0,\cdots, d_{c-1})$,
one of the distances is necessarily equal to $0$ and all the distances are upper
bounded by $c-1$. Vertices which have vector of outer distances $d'$ containing
numbers greater than $c-1$ are associated with the group having a vector $d''$
obtained from $d'$ where all entries greater than $c-1$ are changed to $c-1$. In
this way, there are $t = c^c-(c-1)^c$ groups. We denote the set of all
``normalized'' distance vectors by $D$. Trivially, as $V_d$, $d\in D$, form a
partition of $V$, $\sum_{d\in D}|V_d|=n$.


\section{Bounds on $\alpha$ for existence of cycles}
We first give in Fig.~\ref{fig:non-tree_NE} a construction of a non-tree Nash
equilibrium graph for $n=2s+3$ vertices, and $\alpha=2s = n-3$, for any integer
$s$. This thus shows that the conjecture ``for $\alpha\geq n$, all Nash
equilibria are trees'' cannot be improved to ``for $\alpha \geq
(1-\varepsilon)n$, all Nash equilibria are trees''.
We now proceed and give a lower bound on the length of a shortest cycle in any
Nash equilibrium.

\begin{theorem} 
  \label{thm:no_short_cycle}
  The length $c$ of a shortest cycle $C$ in any Nash equilibrium is at least
  $\frac{2\alpha}{n} + 2$.

\end{theorem}

\begin{proof} 
  We distinguish two cases. First, assume that there is a player, which buys
  both its adjacent edges on the cycle $C$. Without loss of generality assume
  that this player is $a_0$.
  Consider the strategy change where $a_0$ deletes both these edges
  $\{a_0,a_1\}$ and $\{a_0,a_{c-1}\}$ and buys an edge towards player $a_i$ on
  the cycle, $i=2,\ldots,c-2$.
  The player cannot improve by such a change, and therefore ``SAVINGS'' $\leq$
  ``INCREASE''. 
  Here, the player saves at least $\alpha$ (by buying one edge less).
  Let us denote the increase of distances of player $a_0$ to the players of the
  group $V_d$ by $c_{i,d}$. Then we get that $\alpha \leq \sum_{d\in D}
  \delta_{i,d} |V_d|$. Summing up all the $c-3$ inequalities, one for every $i$,
  we get $(c-3)\alpha \leq \sum_{i=2}^{c-2} \sum_d \delta_{i,d} |V_d|$.

  We now show that for every $d$, the coefficient $\sum_i \delta_{i,d}$ at
  $|V_d|$ is at most $(c-2)(c-3)/2$. 
  Consider arbitrary $d=(d_0,d_1,\ldots,d_{c-1})$ of the outer distances of the
  vertices in $V_d$.
  Clearly, the strategy change of $a_0$ increases its distances to $V_d$ iff
  every shortest path from $a_0$ to $V_d$ goes through the deleted edges. 
  Thus, we can assume (for the worst-case) that $d_0=c-1$. 
  Assume that one shortest path (in $G_s$) leaves the cycle at $a_e$,
  $e\in\{1,\ldots,c-2\}$.
  In the new graph $G_{s'}$, player $a_0$ can always use the new edge
  $\{a_0,a_i\}$ and then go to $a_e$ on the remainder of the cycle $C$. 
  Thus, the increase of distances $\delta_{i,d}$ is at most
  $(1+|i-e|)-1=|i-e|$. In total, we obtain $\sum_{i=2}^{c-2} \delta_{i,d} \leq
  \sum_i |i-e| \leq \sum_i (i-1)=(c-3)(c-2)/2$, as claimed.
  Now, since $\sum_{d\in D}|V_d| = n$, we finally get that $\alpha \leq 
  \frac{(c-2)}{2}n$, which gives the claimed $c\geq \frac{2\alpha}{n} + 2$.
  
  Consider now the second case where no player buys two of its adjacent edges in
  $C$, i.e., every player buys exactly one edge. Without loss of generality
  assume that every player $a_i$ buys the edge $\{a_i,a_{i+1}\}$.
  For each player $i$, we consider the strategy change of deleting the edge
  $\{a_i,a_{i+1}\}$. Similarly to the previous case, we obtain $\alpha\leq
  \sum_{d\in D} \delta_{i,d} |V_d|$.
  Summing for every $i$, we get $c\alpha \leq \sum_{i=0}^{c-1} \sum_d
  \delta_{i,d} |V_d|$.
  We show this time that $\sum_{i=0}^{c-1} \delta_{i,d}$, the coefficient at
  $|V_d|$, is upper bounded by $1 + 2 + \dots + (c-2) = (c-2)(c-1)/2$.
  Consider an arbitrary $d=(d_0,\ldots,d_{c-1})\in D$, and assume without loss
  of generality that $d_0=0$. 
  For every player $a_i$, $\delta_{i,d}$ is at most $i-1$, because the
  worst-case increase in a distance of player $a_i$ to vertices $V_d$ happens
  when all shortest paths from $a_i$ used the deleted edge $\{a_i,a_{i+1}\}$.
  But because after the deletion, there is an alternative path from $a_i$ to
  $V_d$ using $a_0$, the increase is at most $i-1$.
  Thus, summing over all $i$, the total increase in distances to $V_d$ is at
  most $0+1+2+\dots+(c-2) = (c-2)(c-1)/2$ as claimed.
  Plugging this into our inequality, $c\alpha \leq \sum_i \sum_d \delta_{i,d}
  |V_d|$ and using the fact that $\sum_d |V_d|=n$, we obtain that
  $c > \frac{2\alpha}{n} + 2$.

\end{proof}

Let $H$ be a non-trivial biconnected component of a non-tree Nash equilibrium,
i.e., an induced subgraph of $H$ of at least three vertices containing no
bridge. For any vertex $v\in H$, let $S(v)$ be the set of vertices which do not
belong to $H$, and which have $v$ as the closest vertex among all vertices in
$H$. For any vertex $u\in H$, we define $deg_H(u)$ to be the degree of vertex
$u$ in the graph induced by $H$. Furthermore, we define $N_k(u)$ to be the
$k$-th neighborhood of $u$ in $H$, i.e., $N_k(u):=\{w\in H\ |\ d(u,w)\leq k\}$.
The following lemma has been shown in \cite{Mihalak+Schlegel/2013}. We will use
it to prove the subsequent lemma.

\begin{lemma}[\cite{Mihalak+Schlegel/2013}]
  \label{lem:no_two_opposite_edges}
  If $u, v \in V (H)$ are two vertices in $H$ with $d(u, v) \geq 3$ such that
  $u$ buys the edge to its adjacent vertex $x$ in a shortest $u-v$-path and $v$
  buys the edge to its adjacent vertex $y$ in that path, then $deg_H (x)
  \geq 3$ or $deg_H (y) \geq 3$.
\end{lemma}


\begin{lemma} 
  \label{lem:5-neighborhood}
  If $H$ is a biconnected component of $G$, then for any vertex $u$, its
  neighborhood $N_5(u)$ in $H$ contains a vertex $v$ with $deg_H(v)\geq
  3$.
\end{lemma}

\begin{proof} 
  Assume that this is not true. Then the 5-neighborhood $N_5(u)$ of vertex $u$
  is formed by two disjoint paths. (The case that the 5-neighborhood forms a
  cycle is excluded by Proposition~\ref{prop:one_cycle} stating that no Nash
  equilibrium for $\alpha>n$ contains exactly one cycle). 
  We consider two cases. First, we will assume that at least one of the two
  paths starting at $u$ is directed away from $u$ (see
  Fig.~\ref{fig:5-Neighborhood}(a)). In the second case, in each
  of the two paths, there has to be a vertex which buys an edge towards $u$.
  It follows from Lemma~\ref{lem:no_two_opposite_edges} that these two vertices
  are the two neighbors of $u$ in $N_5(u)$ (see
  Fig.~\ref{fig:5-Neighborhood}(b)).

  \begin{figure}[t]
    \centering
    \includegraphics[width=0.7\textwidth]{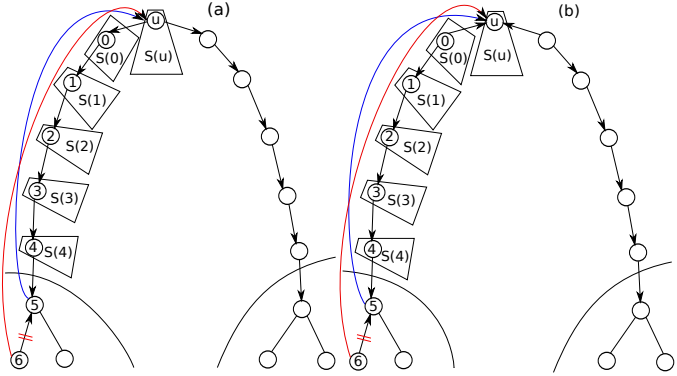}
    \caption{The 5-Neighborhood $N_5(u)$ of vertex $u$.}
    \label{fig:5-Neighborhood}
  \end{figure}

  In the first case, there is a sequence of five edges directed away from $u$,
  with the naming like in Fig.~\ref{fig:5-Neighborhood}(a)).
  Let $s_u:=|S(u)|$, $s_i=|S(i)|$ for $0\leq i\leq 4$.

  Then,
  \begin{equation}
    \label{eq:chain_of_directed_edges}
    s_0\geq s_1+s_2+s_3+s_4, s_1 \geq s_2+s_3+s_4, s_2\geq s_3+s_4, s_3 \geq
    s_4, s_4\geq k\text{,}
  \end{equation}
  where $k$ is the number of vertices which are descendants of vertex $5$ in the
  breadth-first-search (BFS) tree rooted at vertex $3$. We can obtain these
  inequalities by considering the following strategy changes of the players $u$
  and $i$, $0\leq i\leq 3$: delete the edge directed away from $u$, and buy a
  new edge to the next vertex in the sequence; now simply apply the ``SAVINGS''
  $\leq$ ``INCREASE'' principle.

  We first assume that vertex $5$, the neighbor of vertex $4$ in $H$, has degree
  at least 3 in $H$ (i.e., it has at least two children in the BFS tree rooted
  at vertex 3). The case when the degree-3 vertex appears later along the
  path is easier and will be discussed later.
  We now distinguish two cases. First, we assume that one of the children of
  vertex $5$ in the considered BFS tree buys an edge to vertex $5$. Let us call
  it vertex $6$. The other case is when vertex $5$ buys all the edges to its
  children.

  Consider the following strategy change: vertex $6$ deletes an edge towards
  vertex 5 and buys new edge towards vertex $u$. This decreases its distance
  cost at least to vertices in $S(0)$ by 4, and to vertices in $S(1)$ by 2,
  whilst increases distances to vertices in the set of descendants of $5$ in the
  BFS tree rooted at $3$ by at most $6$, to the vertices in $S(4)$ by $4$ and to
  the vertices in $S(3)$ by two. By this strategy change distance from vertex 6 to
  any other vertex is not increased, because vertex $u$ is located deeper than
  vertex $6$ in the BFS tree rooted at vertex $3$.
  But then according to the chain of inequalities
  (\ref{eq:chain_of_directed_edges}) we get $4s_0+2s_1 > 6k+4s_4+2s_3$, and thus
  the player $6$ can improve, a contradiction.

  In the case where vertex 5 buys all edges towards its children, consider the
  following strategy change of vertex 5: delete all the edges to
  its children (in the considered BFS tree) and buy one edge to vertex $u$. By
  this, the ``SAVINGS'' are at least $\alpha$.
  Furthermore, since $H$ is biconnected, the graph remains connected. Distances
  from vertex $5$ are increased only to vertices in the set $K$ -- the set of
  the vertices which are descendants of vertex 5 in the BFS tree rooted at
  vertex 3. This ``INCREASE'' is at most $2\cdot diam(H)$, where $diam(H)$ is
  the diameter of $H$. 
  By the ``SAVINGS'' $\leq$ ``INCREASE'' principle, we get that $\alpha \leq
  2\cdot diam(h) k$.
  At the same time, $\alpha \geq (rad(H)-1) s_0$, where $rad(H)$ is the radius
  of $H$, as otherwise a vertex at distance $rad(H)$ from vertex $0$ could buy
  an edge towards vertex $0$ and decrease its cost.
  Combining these two inequalities with the inequality $s_0\geq 8k$, which is
  obtained from (\ref{eq:chain_of_directed_edges}), we get that
  $8(rad(H)-1)k\leq 2\cdot diam(H)k\leq 4\cdot rad(H)k$, which is a
  contradiction.

  The second case depicted in Fig.~\ref{fig:5-Neighborhood}(b) is analyzed in
  the very same way, the only change is that now the heaviest component is
  $S(u)$. The chain of inequalities is similar to
  (\ref{eq:chain_of_directed_edges}):
  \begin{equation}
    s_u\geq s_0+s_1+s_2+s_3, s_1\geq s_2+s_3+s_4, s_2\geq s_3+s_4, s_3\geq
    s_4, s_4\geq k \text{,}
  \end{equation}
  where the notation is the same as in the first case. We obtain that $s_u\geq
  7k$, and subsequently, arguing about the vertex at distance $rad(H)$ from $u$,
  the contradiction $7(rad(H)-1)k\leq 2\cdot diam(H)k\leq 4 rad(H)k$.

  Finally, if there is a longer sequence of vertices with degree $2$ than the
  considered sequence of length 5 of edges directed away from $u$, then we can
  only consider the last 5 edges (all directed away from $u$) and apply the very
  same reasoning.

\end{proof}

We can strengthen the result if we consider stronger version of a Nash
equilibrium in which no coalition of two players can change their strategies and
improve their overall cost. 

We call such an equilibrium a \emph{$2$-coalitional Nash equilibrium}. 

\begin{lemma} 
  The 3-neighborhood $N_3(u)$ of any vertex $u$ of a biconnected component $H$ 
  of a 2-coalitional Nash equilibrium has a vertex of degree at least 3.
\end{lemma}

\begin{figure}[t]
  \centering
  \includegraphics[width=0.49\textwidth]{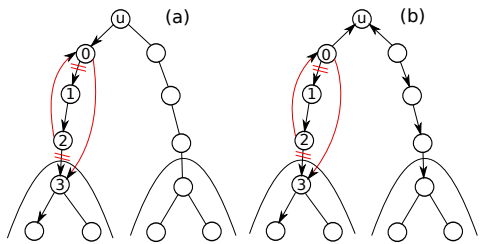}
  \caption{The 3-Neighborhood $N_3 (u)$ of vertex $u$.}
  \label{fig:3-Neighborhood}
\end{figure}

\begin{proof} 
  Assume the converse. Similarly to the proof of Lemma~\ref{lem:5-neighborhood},
  there are two different cases of how the neighborhood of vertex $u$ looks like
  (see Fig.~\ref{fig:3-Neighborhood}(a) and (b); notation is also the same as in
  Lemma~\ref{lem:5-neighborhood}).
  In both cases consider the coalition of players $0$ and $2$. Consider the
  following strategy changes: player $0$ deletes edge $(0,1)$ and instead buys
  edge $(0,3)$, whilst player $2$ deletes edge $(2,3)$ and buys edge $(2,0)$.
  This strategy change does not change the player coalition's creation cost (in
  terms of $\alpha$). Among the vertices $S(0), S(1), S(2)$ and $S(u)$ this
  strategy change decreases coalition's usage cost by $s_u+s_0+s_2$ and
  increases by $s_{1}$.
  Other vertices are partitioned by their shortest distances to vertices $0$ and
  $2$, lets assume that for any vertex $v$ which does not belong to
  $S(0),S(1),S(2)$ or $S(u)$ shortest distance to vertex $0$ is $x$ and shortest
  distance to vertex $2$ is $y$. Obviously $|x-y|\leq 2$. If $|x-y|>0$ then
  there is no increase in the usage cost of coalition towards vertex $v$ by this
  strategy change. The only possibility of increase is when $x=y$, but in that
  case $v$ is the descendant of vertex $3$ in the BFS tree rooted at vertex $1$.
  Similarly to Lemma~\ref{lem:5-neighborhood}, we denote $k$ to be the number of
  vertices which are descendants of vertex $3$ in the BFS tree rooted at vertex
  $1$. Analogously to the proof of Lemma~\ref{lem:5-neighborhood}, the following
  inequalities hold for the case depicted in Fig. \ref{fig:3-Neighborhood}(a):
  $s_0\geq s_1+s_2 , s_1\geq s_2 \mbox{ and } s_2\geq k\text{,}$ whilst for the
  case depicted in Fig.~\ref{fig:3-Neighborhood}(b), we have $s_u\geq
  s_0+s_1+s_2, s_1\geq s_2, s_2\geq k$.
  In both cases $s_u+s_0+s_2>s_1+k$, which results in a contradiction.

\end{proof}

The following two lemmas are crucial for proving the main result of the paper.
The first lemma has been proven in \cite{Mihalak+Schlegel/2013}. The second
lemma strengthens a similar lemma from \cite{Mihalak+Schlegel/2013}. Its proof
uses the result of Theorem~\ref{thm:no_short_cycle}. 

\begin{lemma}[\cite{Mihalak+Schlegel/2013}]
  If the $t$-neighborhood of every vertex of a biconnected component $H$ of a
  Nash equilibrium contains a vertex of degree at least $3$, then the average
  degree of $H$ is at least $2+\frac{1}{3t+1}$.
\end{lemma}

\begin{lemma} 
  If $\alpha > n$, then the average degree of a biconnected component $H$ of a
  Nash equilibrium graph is at most $2+\frac{4n}{\alpha-n}$.
\end{lemma}

\begin{proof} 
  Among all vertices of the equilibrium graph $G$, consider a vertex with the
  smallest usage cost and let this vertex be $v$. Consider a BFS tree $T$ rooted
  in $v$. Let $T' = T\cap H$. Then the average degree of $H$ is
  $deg(H)=\frac{2|E(T')|+2|E(H)\setminus E(T')|}{|V(T')|} \leq
  2+\frac{2|E(H)\setminus E(T')|}{|V(T')|}$. 
  We now bound $|E(H)\setminus E(T')|$. We consider vertices that buy an edge in
  $E(H)\setminus E(T')$ and call them \emph{shopping vertices}. It is easy to
  see that no shopping vertex buys more than 1 edge, because if any of them buys
  two or more edges, it is better for it to delete all of the edges and buy 1 new edge
  towards $v$: this decreases its creation cost by at least $\alpha$, whilst
  increases its usage cost by at most $n$. 
  It is thus enough to bound the number of shopping vertices. For this, we prove
  that the distance in the tree $T'$ between any two shopping vertices is lower
  bounded by $\frac{\alpha-n}{n}$, which then implies that there can not be too
  many shopping vertices. Namely, the number of shopping vertices is at most
  $\frac{2nV(T')}{\alpha-n}$.  
  Assigning every node from $H$ to the closest shopping vertex according to the
  distance in $T'$ forms a partition of $H$, where every part contains exactly
  one shopping vertex. 
  As the distance in $T'$ between shopping vertices is at least
  $\frac{\alpha-n}{n}$, the size of every part is at least
  $\frac{\alpha-n}{2n}$.

  We assume for contradiction that there is a pair of shopping vertices $u_1$
  and $u_2$ such that $d_{T'}(u_1,u_2)<\frac{\alpha-n}{n}$. Let $u_1=x_1,\cdots,
  x_k=u_2$ be the unique path from $u_1$ to $u_2$ in $T'$, and $(u_1,v_1)$ and
  $(u_2,v_2)$ be the edges bought by $u_1$ and $u_2$ in $E(H)\setminus E(T')$. 
  Observe first that vertices $v_1$ and $v_2$ are not descendants of any vertex
  $x_i$, otherwise paths $v_j-x_i$ and $x_i-u_j$ together with an edge
  $(u_j,v_j)$ form a cycle of length at most $2(d_{T'}(u_1,u_2)+1) <
  \frac{2\alpha}{n}+2$ which contradicts Theorem~\ref{thm:no_short_cycle}. 
  Thus, $x_0 := v_1, x_1, \ldots, x_k, x_{k+1}:=v_2$ is a path. Since $x_1$ buys
  edge $(x_0,x_1)$ and $x_k$ buys edge $(x_k,x_{k+1})$, there is a vertex
  $x_{i}$ such that $x_i$ buys both of its adjacent edges $(x_{i-1},x_i)$ and
  $(x_{i},x_{i+1})$. Consider the following strategy change for player $x_i$:
  delete the two adjacent edges and buy a new edge to vertex $v$. In this way
  $x_i$ decreases its creation cost by $\alpha$.

  We now show that $U_\text{new}(x_i)$, the usage cost of $x_i$ in the new graph,
  is less than $U_G(x_i)$, the usage cost in the original graph, plus $\alpha$,
  which gives a contradiction. It is easy to observe that $U_\text{new}(x_i)\leq
  n+U_\text{new}(v)$, since $x_i$ can always go through $v$ in the new strategy
  to any vertex.
  We now consider $U_\text{new}(v)$. Note that only the vertices in the path
  $u_1-u_2$ and their descendants can increase their distance to $v$ by the
  strategy change of $x_i$. Let $y$ be any such vertex. If the closest ancestor
  of $y$ on the path is $x_i$, then $d_\text{new}(v,y)\leq d_G(v,y)$, so there
  is no increase. 
  We assume, without loss of generality, that the closest ancestor (of $y$)
  $x_j$ has an index less than $i$, i.e., $j<i$. Then the following chain of
  inequalities and equalities hold:
  $d_\text{new}(v,y) \leq d_\text{new}(v,x_0) + d_\text{new}(x_0,x_j) +
  d_\text{new}(x_j,y) = d_G(v,x_0) + d_G(x_0,x_j) + d_G(x_j,y)$ (the inequality
  is a triangle inequality, whilst the equality holds because $x_0$ is not a
  descendant of any vertex on the path in the new graph). Since 
  $d_G(v,y) = d_G(v,x_j)+d_G(x_j,y)$, the difference between new and initial
  distances is $d_{new}(v,y) - d_G(v,y) = d_G(v,x_0) + d_G(x_0,x_j) - d_G(v,x_j)
  \leq 2d_G(x_0,x_j) \leq d_G(u_1,u_2) \leq 2\cdot d_{T'}(u_1,u_2) \leq
  \frac{2(\alpha-n)}{n}$ (where the latter inequality is implied by our
  assumption). We need to bound the number of possible $y$'s. Path $u_1-u_2$
  does not go through vertex $v$, so the number of possible $y$'s is bounded by
  the size of the subtree of $T$ of a child of $v$ that contains this path.
  We prove that the size of any subtree of a child of $v$ in $T'$ is at most
  $\frac{n}{2}$.

  Consider any child $t$ of $v$ in $T$, and consider the subtree of $T$ rooted
  in $t$. Let the $b$ be the number of vertices in the subtree, and let $a$ be
  the number of other vertices of $T$. Let $c_1$ be the usage cost of $t$ in the
  subtree, and let $c_2$ be the usage cost of $v$ (!!) in the other part of the
  tree $T$.
  Then the usage cost of $t$ in $G$ is upper bounded by $c_1+a+c_2$, whilst the
  usage cost of $v$ is exactly $b+c_1+c_2$. Since $v$ is the vertex with the
  minimal usage cost, we have $c_1 + a + c_2 \geq b + c_1 + c_2$.
  Since $a + b = n$, we get that $b \leq \frac{n}{2}$.

  Since $y$ was chosen arbitrarily, the increase of the usage cost for $v$ is
  less than $\frac{n}{2}\frac{2(\alpha-n)}{n}=\alpha-n$, and therefore
  $U_\text{new}(v) < U_G(v) + \alpha - n$ which is a contradiction.

\end{proof}

Combining Lemmas $2$ and $3$ with Lemmas $4$ and $5$ gives the main result.

\begin{theorem}
  For $\alpha\geq 65n$ every Nash equilibrium graph is a tree.
\end{theorem}

\begin{theorem}
  For $\alpha \geq 41n$ every 2-coalitional Nash equilibrium graph is a tree.
\end{theorem}

\section{Small cycles and experimental results}

In this section we consider equilibrium graphs that have small girth $c$, and
show that they exist only for small values of $\alpha$. 
We start with an observation that limits the girth of equilibrium graphs
containing exactly one cycle. 

\begin{proposition} 
  \label{prop:one_cycle}
  Let $G$ be a Nash equilibrium graph containing a $k$-cycle $C =
  \{v_0,v_1,\ldots,v_{k-1}\}$, and $F$ the graph where the edges of $C$ are
  removed from $G$. If $F$ consists of $k$ connected components, then $k<6$.
\end{proposition}

\begin{proof}
  Assume for contradiction that $k \geq 6$. For $0 \leq i < k$ let $s_i > 0$
  denote the number of vertices in the connected component of $F$ which contains
  $v_i$. If the edge $(v_0 , v_{k-1})$ is bought by the player $v_0$, then she
  could replace $(v_0, v_{k-1})$ by $(v_0, v_{k-2})$. By doing this, her
  creation cost will remain the same, her distances to $s_{k-3} + s_{k-2}$
  vertices decrease by 1, but her distances to $s_{k-1}$ vertices increase by 1.
  If the edge $(v_0, v_{k-1})$ is bought by the player $v_{k-1}$, this player
  could replace $(v_{k-1} ,v_0)$ by $(v_{k-1}, v_1)$. By this change of her
  strategy, her distances to $s_0$ vertices would increase, but she could
  decrease her distances to $s_1 + s_2$ vertices. 

  Since we consider a Nash equilibrium, we deduce that $s_{k-3} + s_{k-2} \leq
  s_{k-1} \mbox{ or } s_0 \geq s_1 + s_2$. Applying this reasoning for every
  edge of $C$, we get that for every $i$,
  \begin{equation}
    \label{eq:only_one_cycle}
    s_{i-3} + s_{i-2} \geq s_{i-1} \mbox{ or } s_i \geq s_{i+1} + s_{i+2}
    \text{,}
  \end{equation}
  where $0 \leq i < k$ (recall that indexes are considered modulo $c$).
  The two inequalities $s_i \geq s_{i+1} + s_{i+2}$ and $s_{i-1} + s_i \leq
  s_{i+1}$ cannot hold simultaneously. Yet, \ref{eq:only_one_cycle} forces one
  of the inequalities $s_{i-1} +s_i \leq s_{i+1}$ and $s_{i+2} \geq s_{i+3}
  +s_{i+4}$ to be true, so we have that inequality $s_i \geq s_{i+1} + s_{i+2} $
  implies $s_{i+2} \geq s_{i+3} +s_{i+4}$ for any $0\leq i< k$.
  Without loss of generality we can assume that the edge $(v_{k-1},v_0)$ was
  bought by $v_0$. Then we get the chain of inequalities $s_{2i}\geq
  s_{2i+1}+s_{2i+2}$ for every $i$, which is obviously a contradiction.

\end{proof}

We now describe our computer-aided approach for upper-bounding $\alpha$ in case
of an existence of small cycles in Nash equilibrium graphs. In our approach, we
consider a non-tree Nash equilibrium whose smallest cycle has a fixed length
$c$, and we construct a linear program asking for a maximum $\alpha$, whilst
satisfying inequalities of the type ``SAVINGS'' $\leq$ ``INCREASE'', which we
create by considering various strategy changes of the players of the cycle.
The partition of vertices of a Nash equilibrium graph into vertices $V_d$, $d\in
D$, gives a variable $|V_d|$ for every $d$. The number of variables is
$t=c^c-(c-1)^c$. 
We enumerate over all possible (meaningful) directions of the edges on the
considered cycle, and solve the linear program, which gives us an upper bounds
on $\alpha$ for every direction of edges. The largest such value is then
obviously an upper bound on $\alpha$ for any direction, and thus for any Nash
equilibrium containing a cycle of the fixed size. 

The number of all possible directions is equal to $2^c$, but this number can be
decreased to at most $2^{c-3}+2$ by simple observations that all hold without
loss of generality.
We can assume that the number of \emph{right} edges is at least the number of
\emph{left} edges, where an edge $(v_i,v_{i+1})$ is called a \emph{right} edge,
and $(v_{i+1},v_i)$ is called a \emph{left} edge.
Furthermore, we can also assume that the edge $(v_0,v_1)$ is a right edge.
If $c$ is even, every considered cycle can be made (by renaming arguments) to
fall into one of the following three classes: (1) the edges along the cycle
alternate between right and left, or (2) all edges are right edges, or (3) the
first two edges are right edges and the last edge is a left edge.
The same holds when $c$ is odd, with the exception of the alternating edges.

Our linear program contains all inequalities implied by the strategy changes
described in Theorem~\ref{thm:no_short_cycle}. We furthermore add inequalities
for strategy changes of buying one extra edge, and for swapping an edge of the
cycle with a new edge towards an vertex of the cycle. 
We add the equality $\sum_{d\in D}|V_d| = 1$ (which expresses the fact that the
variables should sum up to $n$). Then, the value of a variable $|V_d|$ expresses
the fraction of all vertices (instead of the absolute number of vertices).

We used the GUROBI linear-programming solver to maximize $\alpha$ for every
generated linear program. The largest such value thus gives an upper bound on
$\alpha$ for which a cycle of size $c$ can exist.
Due to the huge number of variables, we could not solve the linear program for
$c>7$, because already for $c=8$, the number of variables was more than $10^7$,
while the number of constraints is $\Theta(c^2)$. 
We have made further tweaks to the code, which allowed us to speed up the
computation.
We observed that many variables had the same coefficients in every generated
constraint, and thus at most one such variable is relevant for obtaining the
solution of the linear program. 
We have considered the variables one by one, and added only those having
unique coefficients in the considered constraints. To check for uniqueness, we
used hashing, as otherwise just creating the matrix of the linear program was
too slow. The obtained compression of the number of variables was huge: for
$c=10$, instead of nearly $10^{10}$ variables we obtained only around $10^5$.

The obtained upper bounds on $\alpha$ are quite close to $n$. For girth $c\leq
7$, we obtain $\alpha \leq 1$, which corresponds to $\alpha\leq n$ if we
required that $\sum_{d\in D}|V_d| = n$ (instead of $\sum_{d\in D}|V_d| = 1$).
For girth $c=8$, $\alpha$ is upper bounded by $\frac{191}{185}$, for girth
$c=9$, $\alpha$ is upper bounded by $\frac{13}{12}$, whilst for
girth $c=10$, $\alpha$ is bounded by $1.2$. 

We have performed further experiments with larger values of $c$, but did not
consider all orientations of edges (as this was out of our computational
power). Furthermore, since the number of variables is increasing
super-exponentially, instead of considering all variables, for larger values of
$c$ we have considered only variables $|V_d|$ that have only $0$'s and $(c-1)$'s
as distances in vector $d$, that is, we have considered $2^c$ variables.
Additionally, we have taken extra $2^c$ random variables.
We have all values of $c$ up to 15. 
Upper bounds for $\alpha$ obtained using only these variables are very close to
the real bounds for $c\leq 10$ (the difference for $k\leq 10$ is between 0 and
0.01).
The largest upper bound of $1.3n$ on $\alpha$ appears for $c=13$, and then only
decreases, which is why we conjecture: the upper-bound of $\alpha \leq 1.3n$ can
be proved by the considered strategy changes.

\vspace{8pt}
\noindent\textbf{Acknowledgements.} This work has been partially supported by
the Swiss National Science Foundation (SNF) under the grant number
200021\_143323/1.

\bibliographystyle{plain}
\bibliography{NetworkCreationGames}

\begin{thebibliography}{1}

\bibitem{Albers+etal/2006}
Susanne Albers, Stefan Eilts, Eyal Even-Dar, Yishay Mansour, and Liam Roditty.
\newblock On {Nash} equilibria for a network creation game.
\newblock In {\em Proc. 17th Annual ACM-SIAM Symposium on Discrete Algorithms
  (SODA)}, pages 89--98, New York, NY, USA, 2006. ACM.

\bibitem{Alon+etal/2013}
Noga Alon, Erik~D Demaine, Mohammad~T Hajiaghayi, and Tom Leighton.
\newblock Basic network creation games.
\newblock {\em SIAM Journal on Discrete Mathematics}, 27(2):656--668, 2013.

\bibitem{Bilo+Guala+Proietti/2012}
Davide Bil\`{o}, Luciano Gual\`{a}, and Guido Proietti.
\newblock Bounded-distance network creation games.
\newblock In {\em Proc. 8th International Workshop on Internet and Network
  Economics (WINE)}, pages 72--85, 2012.

\bibitem{Brautbar+Kearns/2011}
Michael Brautbar and Michael Kearns.
\newblock A clustering coefficient network formation game.
\newblock In {\em Proc. Fourth International Symposium on Algorithmic Game
  Theory (SAGT)}, pages 224--235, 2011.

\bibitem{Demaine+Zadimoghaddam/2010}
Erik Demaine and Morteza Zadimoghaddam.
\newblock Constant price of anarchy in network creation games via public
  service advertising.
\newblock In {\em Proc. Seventh International Workshop on Algorithms and Models
  for the Web-Graph (WAW)}, pages 122--131, 2010.

\bibitem{Demaine+etal/2012}
Erik~D. Demaine, Mohammadtaghi Hajiaghayi, Hamid Mahini, and Morteza
  Zadimoghaddam.
\newblock The price of anarchy in network creation games.
\newblock {\em ACM Trans. Algorithms}, 8(2):1--13, 2012.

\bibitem{Ehsani+etal/2011}
Shayan Ehsani, MohammadAmin Fazli, Abbas Mehrabian, Sina Sadeghian~Sadeghabad,
  MohammadAli Safari, Morteza Saghafian, and Saber ShokatFadaee.
\newblock On a bounded budget network creation game.
\newblock In {\em Proc. 23rd ACM Symposium on Parallelism in Algorithms and
  Architectures (SPAA)}, pages 207--214, 2011.

\bibitem{Fabrikant+etal/2003}
Alex Fabrikant, Ankur Luthra, Elitza Maneva, Christos~H. Papadimitriou, and
  Scott Shenker.
\newblock On a network creation game.
\newblock In {\em Proc. 22nd Annual Symposium on Principles of Distributed
  Computing (PODC)}, pages 347--351, New York, NY, USA, 2003. ACM.

\bibitem{Mihalak+Schlegel/2013}
Mat{\'u}\v{s} Mihal{\'a}k and Jan~Christoph Schlegel.
\newblock The price of anarchy in network creation games is (mostly) constant.
\newblock {\em Theory Comput. Syst.}, 53(1):53--72, 2013.

\end{thebibliography}

\end{document}